\documentclass[12pt]{article}

\usepackage{amsmath}
\usepackage{amsthm}
\usepackage{times}
\usepackage{braket}
\usepackage{dsfont}
\usepackage{graphicx}
\usepackage{hyperref} 
\hypersetup{
    colorlinks=true,
    linkcolor=blue,
    urlcolor=blue,
    citecolor=blue,
}
\usepackage[capitalise]{cleveref}
\usepackage[numbers,sort&compress]{natbib}

\usepackage{etoolbox} 

\def\tr {\mathrm{Tr}}

\newtheorem{theorem}{Theorem}

\newtheorem{corollary}{Corollary}
\newtheorem{lemma}{Lemma}

\title{Generalized Wigner-Yanase Skew Information and the Affiliated Inequality}

\author
{Ma-Cheng Yang$^{1}$, Cong-Feng Qiao$^{1,2}$\footnote{correspondence: qiaocf@ucas.ac.cn}\\
\\
\normalsize{$^{1}$School of Physical Sciences, University of Chinese Academy of Sciences}\\
\normalsize{YuQuan Road 19A, Beijing 100049, China}\\
\normalsize{$^{2}$Key Laboratory of Vacuum Physics, University of Chinese Academy of Sciences}\\
\normalsize{YuQuan Road 19A, Beijing 100049, China}\\
\\
}

\date{}

\begin{document}


\baselineskip24pt

\maketitle



\begin{abstract}
A family of skew information quantities is obtained, in which the well-known Wigner-Yanase skew information and quantum Fisher information stand as special cases. A transparent proof of convexity of the generalized skew information is given, implying a simple proof of the Wigner-Yanase-Dyson conjecture. We find in this work an exact skew information inequality for qubit system, which may regard as the information counterpart of the uncertainty relation. A lower bound for generalized skew information of a pair of incompatible observables in arbitrary dimension and also the upper bound for qubit system are achieved.
\end{abstract}






\section{Introduction}

Uncertainty relation exhibits the intrinsic limit of simultaneous measurement of incompatible observables, which in fact captures the unpredictability of the quantum world. In other words, one is hindered to gain enough information by eliminating as much as possible uncertainties in the process of quantum measurement. Based on the measures of variance, entropy, etc., abundant uncertainty relations have been established, see e.g. Refs. \cite{heisenberg27,robertson29,schrodinger30,maassen88,friedland13,coles17,li19,li20-tvh34}, which unveil the measurement limit from various different aspects. On the other hand, considering of the indeterminacy inhabited in microworld, it is interesting to know the quantum limit in information content. In order to quantify the information content of a quantum state $\rho$ with respect to an observable $X$, Wigner and Yanase defined the measure of ``skew information" \cite{wigner63}. It writes
\begin{align}
I(\rho,X) = -\frac{1}{2}\tr\left[[\sqrt{\rho},X]^2\right] \; .
\label{si}
\end{align}
Here, $[A,B]=AB-BA$ denotes commutator.

Equipped with the definition on $I(\rho,X)$ in (\ref{si}), it is tempting to formulate certain constraint similar to uncertainty relation to determine the accessibility of information in the measurement of incompatible observables. In Ref. \cite{luo03-exuie}, Luo constructed an uncertainty relation based on Wigner-Yanase skew information, the Theorem 2 therein. However, as pointed out in Refs. \cite{rivas08,hansen08,li09}, that theorem may have a problem and a counterexample was indeed given \cite{rivas08}. Later on, an even tighter inequality \cite{luo04} was obtained, and again it was also found beatable \cite{yanagi05}. Up to now, unfortunately, these problems still have no proper solutions yet. Besides, it is notable that in regard of the information content, an upper bound for skew information product of two incompatible observables is in some sense even meaningful, which tells us maximally how much information one may obtain.

In this work, we introduce a family of measures about the amount of information, which satisfies the definition given by Wigner and Yanase \cite{wigner63}. Furthermore, the Wigner-Yanase skew information and quantum Fisher information belong to this family. It is worthy noting that the Wigner-Yanase skew information and quantum Fisher information have been widely employed to the research of quantum information, such as uncertainty relation \cite{dehesa07,gibilisco07,andai08,gibilisco08,toranzo15,zhang21,ren21-skew}, entanglement detection \cite{chen05,li13,hyllus12,akbari-kourbolagh19,hong15,ren21}, quantum coherence \cite{girolami14,karpat14,luo17}, and so forth. It is highly expected that the generalized skew information may impact on the quantum information science. We prove the convexity of the generalized skew information by means of a plain technique, which implies the correctness of Wigner-Yanase-Dyson conjecture. It should be noted that the conjecture had ever been proved by Lieb \cite{lieb73,nielsen10} and Hansen later \cite{hansen06,hansen08}, whereas subject to a sophisticated proof of the theorem therein. We obtain a lower bound for the generalized skew information of a pair of incompatible observables in arbitrary dimension, and in the meantime an upper bound for the qubit system.

\section{The generalized Wigner-Yanase skew information}

We use $M_{d}(\mathds{C})$ to denote the vector space constructed by $d\times d$ complex matrices and define the following function on $M_{d}(\mathds{C})$
\begin{align}
\zeta_{\rho}(X,Y) &:= \tr[\rho X^{\dagger}Y] - \notag \\ &\sum_{i,j}f(\lambda_{i},\lambda_{j})\braket{\psi_{i}|X^{\dagger}|\psi_{j}}\braket{\psi_{j}|Y|\psi_{i}} \; .
\end{align}
Here, $\rho=\sum_{i=1}^{d}\lambda_{i}\ket{\psi_{i}}\bra{\psi_{i}}$, $\lambda_{r+1}=\cdots=\lambda_{d}=0$ and $\operatorname{rank}(\rho)=r$. $\ket{\psi_{r+1}},\cdots,\ket{\psi_{d}}$ are orthogonal bases of the null space of $\rho$, which implies the completeness relation $\sum_{i=1}^{d}\ket{\psi_{i}}\bra{\psi_{i}}=\mathds{1}$. And the eigenvalues $\lambda_{i}$s are arrayed in descending order, i.e., $\lambda_{1}\geq \cdots\geq\lambda_{r}>0$. Then, it is straightforward to check that
$\zeta_{\rho}(X,Y)$ is a sesquilinear function on the $M_{d}(\mathds{C})$.

In this work we will focus on a special representative case $f(\lambda_{i},\lambda_{j})=m_{s}(\lambda_{i},\lambda_{j})$ with $m_{s}(a_{1},a_{2})$ denoting the generalized mean of arbitrary positive real numbers $a_{1}$ and $a_{2}$ of equal weighting \cite{hardy52,bullen03}
\begin{align}
m_{s}(a_{1},a_{2}) = \left(\frac{a_{1}^{s}+a_{2}^{s}}{2}\right)^{1/s} \; ,
\label{eq:gene_mean}
\end{align}
where $-\infty<s<0$. For cases $s=0$ and $-\infty$, $m_{s}(a_{1},a_{2})$ is defined by limiting processes \cite{hardy52}
\begin{align}
&m_{0}(a_{1},a_{2}):=\lim_{s\rightarrow 0}m_{s}(a_{1},a_{2})=\sqrt{a_{1}a_{2}} \; , \\
&m_{-\infty}(a_{1},a_{2}):=\lim_{s\rightarrow -\infty}m_{s}(a_{1},a_{2})=\min\{a_{1},a_{2}\} \; .
\end{align}
To ensure \cref{eq:gene_mean} being well-defined, we require $m_{s}(a_{1},0)=m_{s}(0,a_{2})=m_{s}(0,0)=0$.
Then we can define the following sesquilinear function on $M_{d}(\mathds{C})$
\begin{align}
\zeta_{\rho}^{s}&(X,Y) := \tr[\rho X^{\dagger}Y] - \notag \\ &\quad \sum_{i,j}m_{s}(\lambda_{i},\lambda_{j})\braket{\psi_{i}|X^{\dagger}|\psi_{j}}\braket{\psi_{j}|Y|\psi_{i}} \; , \notag \\
&= \sum_{i\neq j}\left[\lambda_{i}-m_{s}(\lambda_{i},\lambda_{j})\right]\braket{\psi_{i}|X^{\dagger}|\psi_{j}}\braket{\psi_{j}|Y|\psi_{i}} \; .
\label{zetarho}
\end{align}
Here, $M_{d}(\mathds{C})$ represents the dimension-$d^2$ complex inner product space \cite{bognar74}. Due to the hermiticity, $\zeta_{\rho}^{s}(X,X)=\zeta_{\rho}^{s}(X,X)^{*}$, there are three possible cases, i.e. $\zeta_{\rho}^{s}(X,X)>0$, $\zeta_{\rho}^{s}(X,X)<0$, or $\zeta_{\rho}^{s}(X,X)=0$. Correspondingly, $X$ is said to be positive, negative, or neutral vector. In quantum mechanics, $X$ denotes an observable.

In the following, we will prove that $\zeta_{\rho}^{s}(X,X)$ is always nonnegative, hence $\zeta_{\rho}^{s}(X,X)$ can be treated as the generalized Wigner-Yanase skew information
\begin{align}
I^{s}(\rho,X) :=& \zeta_{\rho}^{s}(X,X) \; .
\end{align}
Indeed, the Wigner-Yanase skew information is merely the special case of $s=0$
\begin{align}
I(\rho,X) &= I^{0}(\rho,X) \notag \\
&= \tr[\rho X^{\dagger}X] - \sum_{i,j}m_{0}(\lambda_{i},\lambda_{j})|\braket{\psi_{i}|X|\psi_{j}}|^2 \\
&= \tr[\rho X^{\dagger}X] - \tr[\sqrt{\rho}X^{\dagger}\sqrt{\rho}X] \\
&= -\frac{1}{2}\tr\left[[\sqrt{\rho},X]^2\right] \; .
\end{align}
Here, $m_{0}(a_{1},a_{2})=\sqrt{a_{1}a_{2}}$ signifies the geometric mean. Notice that if $m_{s}(\lambda_{i},\lambda_{j})$ is replaced by the generalized mean with weignt $w\in[0,1]$, i.e. $m_{s}(\lambda_{i},\lambda_{j})=\left(w \lambda_{i}^{s}+(1-w)\lambda_{j}^{s}\right)^{1/s}$, we obtain the Wigner-Yanase-Dyson skew information $I_{w}(\rho,X)=\tr[\rho X^{\dagger}X]-\tr[\rho^{w}X^{\dagger}\rho^{1-w}X]$ due to the fact $\lim_{s\rightarrow 0}m_{s}(\lambda_{i},\lambda_{j})=\lambda_{i}^{w}\lambda_{j}^{(1-w)}$. It is easy to see that the quantum Fisher information (QFI) \cite{helstrom69,toth14} is the special case of $s=-1$
\begin{align}
F(\rho,X) &= I^{-1}(\rho,X) \\
= \tr&[\rho X^{\dagger}X] - \sum_{i,j}m_{-1}(\lambda_{i},\lambda_{j})|\braket{\psi_{i}|X^{\dagger}|\psi_{j}}|^2 \; .
\end{align}
Here, $m_{-1}(a_{1},a_{2})=\frac{2}{1/a_{1}+1/a_{2}}$ signifies harmonic mean. Given $X$ an observable, $X=X^{\dagger}$, QFI can be expressed in terms of a symmetric logarithmic derivative \cite{braunstein94,holevo11,toth14}, that is $F(\rho,X)=\frac{1}{4}\tr[{\rho L^2}]$ with $L=2i\sum_{k,l=1}^{r}\frac{\lambda_k-\lambda_l} {\lambda_k+\lambda_l}\braket{\psi_{k}|X|\psi_{l}}\ket{\psi_{k}}\bra{\psi_{l}}$. It is noted that Hansen has proposed the metric adjusted skew information as a generalized Wigner-Yanase skew information via Morozova-Chentsov function \cite{hansen08}. Because the generalized mean $m_{s}(a_1,a_2)$ is an operator mean, that is, $m_{s}(1,a)$ is an operator monotone function and satisfying equation $m_{s}(1,a)=am_{s}(1,a^{-1})$, if and only if $-1\leq s\leq 1$ \cite{nakamura89}, $I^{s}(\rho,X)$ reduces to the metric adjusted skew information if $-1\leq s\leq 0$.

As a measure of information amount, $I^{s}(\rho,X)$ should be nonnegative and ``skew'' implies $I^{s}(\rho,X)=0$ if $[\rho,X]=0$, as argued by Wigner and Yanase \cite{wigner63}. In forthcoming contents, we have a detailed discussions and prove that the generalized skew information $I^{s}(\rho,X)$ satisfies the same requirements with Wigner-Yanase skew information.

\begin{theorem}
$I^{s}(\rho,X)$ is independent of the trace of $X$ and satisfies $I^{s}(U\rho U^{\dagger},X)=I^{s}(\rho,U^{\dagger}XU)$ for some unitary matrix $U$; If $X$ is a conserved quantity of an isolated system, $I^{s}(\rho,X)$ is independent of time; $I^{s}(\rho,X) \geq 0$, where equal holds if and only if $\rho$ and $X$ are commutative, $[\rho,X]=0$.
\label{th:gene_skew_non}
\end{theorem}
\begin{proof}
By definition, the generalized skew information can be expressed as
\begin{align}
I^{s}(\rho,X) = \sum_{i\neq j}\left[\lambda_{i}-m_{s}(\lambda_{i},\lambda_{j})\right]\left|\braket{\psi_{i}|X^{\dagger}|\psi_{j}}\right|^2 \; ,
\end{align}
which is not dependent on the diagonal entries of $X$ in the eigenspace of $\rho$. Considering that the trace of $X$ is independent of basis, $I^{s}(\rho,X)$ is independent of the trace of $X$. Via above equation, obviously $I^{s}(\rho,X)$ satisfies $I^{s}(U\rho U^{\dagger},X)=I^{s}(\rho,U^{\dagger}XU)$ for some unitary matrix $U$. If $X$ is a conserved quantity of an isolated system, then we have $[U,X]=0$ and $I^{s}(U\rho U^{\dagger},X)=I^{s}(\rho,X)$, where $U$ is the time revolution operator. As a matter of fact, we can prove that $I^{s}(\rho,X) \geq 0$ is true even for normal matrix $X$. If $X$ is normal, i.e. $XX^{\dagger}=X^{\dagger}X$, then $\tr[\rho X^{\dagger}X]=\sum_{ij}\lambda_{i}|\braket{\psi_{i}|X^{\dagger}|\psi_{j}}|^{2}$, $\tr[\rho XX^{\dagger}]=\sum_{ij}\lambda_{j}|\braket{\psi_{i}|X^{\dagger}|\psi_{j}}|^{2}$ and $\tr[\rho X^{\dagger}X]=\tr[\rho XX^{\dagger}]=\sum_{ij}\frac{\lambda_{i}+\lambda_{j}}{2}|\braket{\psi_{i}|X^{\dagger}|\psi_{j}}|^{2}$. Thus
\begin{align}
&I^{s}(\rho,X)= \sum_{i,j}\left[\frac{\lambda_{i}+\lambda_{j}}{2}-m_{s}(\lambda_{i},\lambda_{j})\right]|\braket{\psi_{i}|X^{\dagger}|\psi_{j}}|^2\geq 0 \; , \notag
\end{align}
due to the monotonicity of the generalized mean i.e., $m_{r}(a_{1},a_{2})\leq m_{s}(a_{1},a_{2})$ if $r<s$ \cite{hardy52}. If $[\rho,X]=0$, $X$ is then diagnoseable in the eigenspace of $\rho$. Therefore, $I^{s}(\rho,X)=\sum_{i\neq j}\left[\lambda_{i}-m_{s}(\lambda_{i},\lambda_{j})\right]|\braket{\psi_{i}|X^{\dagger}|\psi_{j}}|^2=0$ and vice versa.
\end{proof}
\begin{theorem}\label{th:gene_skew_mono}
For arbitrary observable $X$ and quantum state $\rho$, $I^{s}(\rho,X)$ is a monotonically decreasing function of $s$ and less than variance, say
\begin{align}
I^{0}(\rho,X) \leq \cdots \leq I^{-\infty}(\rho,X) \leq V(\rho,X) \; ,
\end{align}
and $I^{s}(\rho,X)=V(\rho,X)$ when $\rho$ is a pure state. Here, $V(\rho,X)=\tr[\rho X^{\dagger}X]-|\tr[\rho X^{\dagger}]|^2$ stands for variance.
\end{theorem}
\begin{proof}
Since the generalized mean $m_{s}(a_{1},a_{2})$ monotonically increases with $s$ \cite{hardy52}, $I^{s}(\rho,X)$ is a monotonically decreasing function. Hence, according to the definitions of $I^{s}(\rho,X)$ and $V(\rho,X)$, we only need to compare $\sum_{i,j}m_{s}(\lambda_{i},\lambda_{j})|\braket{\psi_{i}|X^{\dagger}|\psi_{j}}|^2$ with $|\tr[X^{\dagger}]|^2$. It is obvious that $g_{\rho}(X,Y)=\sum_{i,j}m_{s}(\lambda_{i},\lambda_{j})$ $\braket{\psi_{i}|X^{\dagger}|\psi_{j}}\braket{\psi_{j}|Y|\psi_{i}}$ is a sesquilinear function on the $M_{d}(\mathds{C})$ and $g_{\rho}(X,X)=\sum_{i,j}m_{s}(\lambda_{i},\lambda_{j})$ $|\braket{\psi_{i}|X^{\dagger}|\psi_{j}}|^2\geq 0, \forall X\in M_{d}(\mathds{C})$. Thus, $g_{\rho}(X,Y)$ is a semidefinite inner product on the $M_{d}(\mathds{C})$, which implies the Cauchy-Schwarz inequality, $\forall X,Y \in M_{d}(\mathds{C})$,
\begin{align}
g_{\rho}(X,X)g_{\rho}(Y,Y) \geq |g_{\rho}(X,Y)|^2 \; .
\end{align}
Taking $X$ to be a hermitian operator and $Y=\mathds{1}$, we have
\begin{align}
g_{\rho}(X,X) \geq \left(\sum_{i}\lambda_{i}\braket{\psi_{i}|X|\psi_{i}}\right)^2 = \tr[\rho X]^2 \; ,
\end{align}
where $g_{\rho}(\mathds{1},\mathds{1})=\sum_{i,j}m_{s}(\lambda_{i},\lambda_{j})\delta_{ij}=\sum_{i}\lambda_{i}=1$ is employed. Therefore the conclusion $I^{s}(\rho,X) \leq V(\rho,X)$ is obtained. If $\rho=\ket{\psi}\bra{\psi}$, the pure state, then $\sum_{i,j}m_{s}(\lambda_{i},\lambda_{j})|\braket{\psi_{i}|X|\psi_{j}}|^2=(\braket{\psi|X|\psi})^2$, and hence $I^{s}(\rho,X)=V(\rho,X)$.
\end{proof}

Since the unification of ensembles will inevitably leads to certain information loss of individual ensembles, the information content of a grand ensemble is less than that of the average information content of component ensembles. Therefore, the generalized skew information should satisfy the following convexity.
\begin{theorem}
For arbitrary observable $X$, $I^{s}(\rho,X)$ is a convex function of density matrices, i.e.,
\begin{align}
I^{s}\left(\rho,X\right) \leq \sum_{i}p_{i}I^{s}(\rho_{i},X) \; .
\end{align}
Here, $\rho=\sum_{i}p_{i}\rho_{i}$ with $\sum_{i}p_{i}=1$ and $0<p_{i}<1$.
\end{theorem}
\begin{proof}
Because the set of density matrices is a convex hull of pure states, any density matrix can then be expressed as the convex combination of pure sates,
\begin{align}
\rho = \sum_{i}p_{i}\ket{\phi_{i}}\bra{\phi_{i}} \; .
\label{eq:ensemble_decom}
\end{align}
Considering of this, without loss of generality, we are legitimate to perform the proof with the convex combination of pure states. Since $\tr[\rho X^2]$ is linear in density matrix, it is enough only to compare $\sum_{i,j}m_{s}(\lambda_{i},\lambda_{j})|\braket{\psi_{i}|X|\psi_{j}}|^2$ with $\sum_{i}p_{i}|\braket{\phi_{i}|X|\phi_{i}}|^2$. The decomposition of \cref{eq:ensemble_decom} is not unique, there exist infinite number of similar decompositions. This can be grasped in light of the well-known theorem for the classification of quantum ensembles \cite{hughston93,nielsen10}, which states that ensembles $\{p_{i},\ket{\phi_{i}}\}$ and $\{q_{j},\ket{\varphi_{j}}\}$ generate the same density matrix if and only if $\sqrt{p_{i}}\ket{\phi_{i}}=\sum_{j}U_{ij}\sqrt{q_{j}}\ket{\varphi_{j}}$ hold for some unitary matrix $U$. Thus, the ensemble $\{p_{i},\ket{\phi_{i}}\}$ is associated with eigen-ensemble $\{\lambda_{i},\ket{\psi_{i}}\}$ via unitary matrix $U$, i.e. $\sqrt{p_{i}}\ket{\phi_{i}} = \sum_{j}U_{ij}\sqrt{\lambda_{j}}\ket{\psi_{j}}$. Then
\begin{align}
\sum_{i}p_{i}|\braket{\phi_{i}|X|\phi_{i}}|^2 &= \sum_{i}\left|\sum_{j,j'}U_{ij'}^{*}X'_{j'j}U_{ji}^{\mathrm{T}}\right|^2 \nonumber \\
&= \sum_{i}|X''_{ii}|^2 \leq \tr[X'^2] \nonumber \\
&= \sum_{i,j}\lambda_{i}\lambda_{j}|\braket{\psi_{i}|X|\psi_{j}}|^2 \; .
\end{align}
Here, $X'_{j'j}=\sqrt{\lambda_{j'}\lambda_{j}}\braket{\psi_{j'}|X|\psi_{j}}$, $X''=U^{*}X'U^{\mathrm{T}}$ and $\tr[X'^2]=\tr[X''^2]$ since $U$ is unitary. Considering
\begin{align}
\lambda_{i}\lambda_{j}\leq \min\{\lambda_{i},\lambda_{j}\} \; , \text{for} \; 0\leq\lambda_{i},\lambda_{j}\leq 1 \; ,
\end{align}
then the assertion follows.
\end{proof}
We have proved the convexity of the generalized skew information by virtue of the theorem for classification of quantum ensembles. If $m_{s}(\lambda_{i},\lambda_{j})$ is replaced by the generalized mean with weight $w\in[0,1]$, i.e., $m_{s}(\lambda_{i},\lambda_{j})=\left(w \lambda_{i}^{s}+(1-w)\lambda_{j}^{s}\right)^{1/s}$, we can then obtain more skew information quantities, and the convexity can be proved in a similar way. This indicates that the renowned Wigner-Yanase-Dyson conjecture is confirmed.

\begin{theorem}
The information content of the union of two systems should be the sum of
the information contents of the components, that is, the generalized skew information satisfies additivity
\begin{align}
I^{s}(\rho_{A}\otimes\rho_{B},X_{A}+X_{B}) = I^{s}(\rho_{A},X_{A}) +I^{s}(\rho_{B},X_{B}) \; ,
\end{align}
where $X_{A}$ and $X_{B}$ are observables of two subsystems respectively.
\end{theorem}
\begin{proof}
We only need to calculate the two terms $\tr[\rho X^2]$ and $\sum_{i,j}m_{s}(\lambda_{i},\lambda_{j})|\braket{\psi_{i}|X|\psi_{j}}|^2$ respectively. After simplication, we have $\tr[\rho X^2]=\tr[(\rho_{A}\otimes\rho_{B})(X_{A}+X_{B})^2] = \tr[\rho_{A}X_{A}^2] + \tr[\rho_{B}X_{B}^2] +2\tr[\rho_{A}X_{A}]\tr[\rho_{B}X_{B}]$. Next we calculate $\sum_{i,j}m_{s}(\lambda_{i},\lambda_{j})|\braket{\psi_{i}|X|\psi_{j}}|^2$. Assuming that $\rho_{A},\rho_{B}$ have the following spectrum decompositions
\begin{align}
\rho_{A} = \sum_{i}\lambda_{i}^{A}\ket{\psi_{i}^{A}}\bra{\psi_{i}^{A}} \; , \; \rho_{B}=\sum_{i}\lambda_{i}^{B}\ket{\psi_{i}^{B}}\bra{\psi_{i}^{B}} \; .
\end{align}
And then $\rho=\rho_{A}\otimes\rho_{B}$ has spectrum decomposition
\begin{align}
\rho = \sum_{i,j}\lambda_{i}^{A}\lambda_{j}^{B}\ket{\psi_{i}^{A}}\ket{\psi_{j}^{B}}\bra{\psi_{i}^{A}}\bra{\psi_{j}^{B}} \; .
\end{align}
The second term becomes
\begin{align}
&\sum_{ij}\sum_{kl}m_{s}(\lambda_{i}^{A}\lambda_{j}^{B},\lambda_{k}^{A}\lambda_{l}^{B})\left|\bra{\psi_{i}^{A}}\bra{\psi_{j}^{B}}X_{A}+X_{B}\ket{\psi_{k}^{A}}\ket{\psi_{l}^{B}}\right|^2 = \notag \\
&\sum_{ij}\sum_{kl}m_{s}(\lambda_{i}^{A}\lambda_{j}^{B},\lambda_{k}^{A}\lambda_{l}^{B})\left|\bra{\psi_{i}^{A}}X_{A}\ket{\psi_{k}^{A}}\delta_{jl}+\bra{\psi_{j}^{B}}X_{B}\ket{\psi_{l}^{B}}\delta_{ik}\right|^2 \; ,
\end{align}
which simplifies into three terms
\begin{align}
\sum_{ij}\sum_{k}m_{s}(\lambda_{i}^{A}\lambda_{j}^{B},\lambda_{k}^{A}\lambda_{j}^{B})\left|\braket{\psi_{i}^{A}|X_{A}|\psi_{k}^{A}}\right|^2 \; , \\ \sum_{ij}\sum_{l}m_{s}(\lambda_{i}^{A}\lambda_{j}^{B},\lambda_{i}^{A}\lambda_{l}^{B})\left|\braket{\psi_{j}^{B}|X_{B}|\psi_{l}^{B}}\right|^2 \; , \\
2\sum_{ij}m_{s}(\lambda_{i}^{A}\lambda_{j}^{B},\lambda_{i}^{A}\lambda_{j}^{B})\bra{\psi_{i}^{A}}X_{A}\ket{\psi_{i}^{A}}\bra{\psi_{j}^{B}}X_{B}\ket{\psi_{j}^{B}} \; .
\end{align}
Considering that the properties of generalized mean $m_{s}(a_{1}b,a_{2}b)=bm_{s}(a_{1},a_{2})$ and $m_{s}(a,a)=a$, these terms become
\begin{align}
\sum_{ik}m_{s}(\lambda_{i}^{A},\lambda_{k}^{A})\left|\braket{\psi_{i}^{A}|X_{A}|\psi_{k}^{A}}\right|^2 \; , \\
\sum_{jl}m_{s}(\lambda_{j}^{B},\lambda_{l}^{B})\left|\braket{\psi_{j}^{B}|X_{B}|\psi_{l}^{B}}\right|^2 \; , \\
2\tr[\rho_{A}X_{A}]\tr[\rho_{B}X_{B}] \; .
\end{align}
Via $\tr[\rho X^2]-\sum_{i,j}m_{s}(\lambda_{i},\lambda_{j})|\braket{\psi_{i}|X|\psi_{j}}|^2$, the assertion follows.
\end{proof}
Above, we prove all the properties proposed by Wigner and Yanase besides the subadditivity
\begin{align}
I^{s}(\rho,X_{A}+X_{B}) \geq I^{s}(\rho_{A},X_{A}) +I^{s}(\rho_{B},X_{B}) \; ,
\end{align}
where $\rho_{A}=\tr_{B}[\rho],\rho_{B}=\tr_{A}[\rho]$ are reduced density matrices. Wigner and Yanase only proved that the subadditivity is true for pure state and conjectured that it is always true because the whole system contains the additional statistical correlation information. Unfortunately, the conjecture has been negatived by Hansen \cite{hansen07}. However, quantum state $\rho$ always can be transformed into the normal form by local filtering or SLOCC (stochastic local operations assisted by classical communication) transformation \cite{verstraete03}
\begin{align}
\rho = \frac{1}{d^2}\mathds{1}\otimes\mathds{1} + \sum_{\mu\nu}\chi_{\mu\nu}\pi_{\mu}\otimes\pi_{\nu} \; ,
\end{align}
where $\pi_{\mu}$s are traceless orthogonal
observables. Obviously, under the local filtering transformation all the quantum state will satisfy subadditivity because the reduced matrices are the maximal mixed state and skew information will be zero.

The generalized skew information $I^{s}(\rho,X)$ reduces to the variance for pure state. The next result reveals that the generalized skew information only depends on the purity of quantum state in some sense for a qubit system.
\begin{theorem}\label{th:qubit_gene_skew}
For a qubit system, if quantum sates $\rho$ and $\rho'$ have same purity, i.e. $\tr[\rho^{2}]=\tr[\rho'^{2}]$, then for nonzero observables $X$ and $X'$, in case $[\rho', X']\neq 0$ we have
\begin{align}
I^{s}(\rho,X) = \eta I^{s}(\rho',X') \; ,
\end{align}
where $\eta=|\braket{\psi_{1}|X|\psi_{2}}|^2/|\braket{\psi'_{1}|X'|\psi'_{2}}|^2$.
\end{theorem}
To prove the claim, we first show the following lemma.
\begin{lemma}
For a qubit system, if $\rho$ and $\rho'$ have the same purity, they possess the same set of eigenvalues.
\label{lem:qubit_purity}
\end{lemma}
\begin{proof}
Assume the eigenvalues of $\rho,\rho'$ satisfy $\lambda_{1}\geq \lambda_{2}$ and $\lambda'_{1}\geq \lambda'_{2}$ respectively. $\tr[\rho^2]=\tr[\rho'^2]$ implies $\lambda_{1}^2+\lambda_{2}^2=\lambda'^2_{1}+\lambda'^2_{2}$. Considering $\lambda_{1}+\lambda_{2}=\lambda'_{1}+\lambda'_{2}=1$, it is easy to find $\lambda_{1}\lambda_{2}=\lambda'_{1}\lambda'_{2}$, and therefore  $(\lambda_{1}-\lambda'_{1})(1-\lambda_{1}-\lambda'_{1})=0$, which tells $\lambda_{1}=\lambda'_{1}$ on account of $\lambda_{1}\geq \lambda_{2}$ and $\lambda'_{1}\geq \lambda'_{2}$.
\end{proof}
For a qubit system, we notice that
\begin{align}
I^{s}(\rho,X) =& (\lambda_{1}-m_{s}(\lambda_{1},\lambda_{2}))|\braket{\psi_{1}|X|\psi_{2}}|^2 + \notag \\
&(\lambda_{2}-m_{s}(\lambda_{2},\lambda_{1}))|\braket{\psi_{2}|X|\psi_{1}}|^2 \notag \\
=& (1-2m_{s}(\lambda_{1},\lambda_{2}))|\braket{\psi_{1}|X|\psi_{2}}|^2 \; .
\end{align}
According to \cref{lem:qubit_purity}, we can find $I^{s}(\rho,X)=\eta I^{s}(\rho',X')$ with $\eta=|\braket{\psi_{1}|X|\psi_{2}}|^2/|\braket{\psi'_{1}|X'|\psi'_{2}}|^2$.
\begin{figure}
\centering
\includegraphics[width=0.8\linewidth]{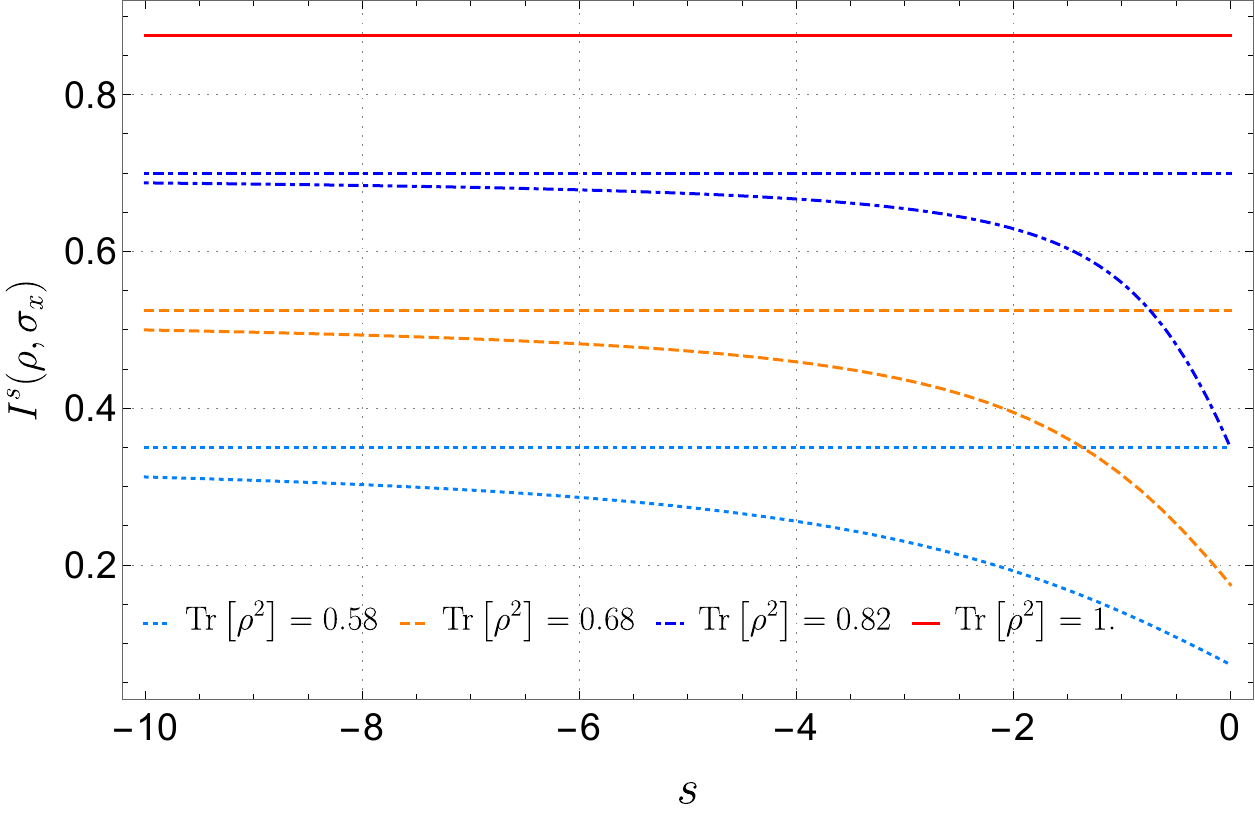}
\caption{\label{fig:qubit_gene_skew_info} The generalized skew information $I^{s}(\rho,X)$ with different purity. The horizontal lines exhibit $I^{-\infty}(\rho,X)$ for different purity, while the other curves show the changes of $I^{s}(\rho,X)$ versus $s$ with definite purity. Here, $\rho=1/2(\mathds{1}+\vec{r}\cdot\vec{\sigma})$, $\tr[\rho^2]=1/2(1+|\vec{r}|^2)$. Note, the solid red curve and solid red line are overlapped for pure states.}
\end{figure}
In \cref{fig:qubit_gene_skew_info}, we plot $I^{s}(\rho,\sigma_{x})$ versus different purity for qubit system, which visualizes the conclusions of \cref{th:gene_skew_non,th:gene_skew_mono}. Note, here for qubit system it is enough to consider only the Pauli operator $\sigma_{x}$ for exhibition according to \cref{th:qubit_gene_skew}.

\section{The skew information inequality}

Uncertainty relation sets a lower bound for measurement of a pair of observables $X,Y$, which mathematically writes
\begin{align}
V(\rho,X)V(\rho,Y) \geq |\operatorname{Cov}(\rho;X,Y)|^2 \; ,
\label{eq:schrodinger_form}
\end{align}
the Schrödinger form \cite{schrodinger30}. Here, $\operatorname{Cov}(\rho;X,Y)$ $=$ $\tr[\rho X^{\dagger}Y]-\tr[\rho X^{\dagger}]\tr[\rho Y]$ signifies the covariance. The skew information also has a similar expression which reveals the quantum limit of information content.
\begin{theorem}
For arbitrary dimensional system, a skew information inequality exists, that is
\begin{align}
I^{s}(\rho,X)I^{s}(\rho,Y) \geq \frac{1}{16}(I^{s}(\rho,X+Y)-I^{s}(\rho,X-Y))^2 \; .
\end{align}
\label{theo:gene_qfi_ineq}
\end{theorem}
\begin{proof}
We first show that $\operatorname{Re} \zeta_{\rho}^{s}(X,Y)$, defined in (\ref{zetarho}), is a bilinear function in the real vector space $\mathcal{H}_{d}=\{X|X=X^{\dagger},X\in M_{d}(\mathds{C})\}$. This is somewhat obvious, since one may notice that $\operatorname{Re} \zeta_{\rho}^{s}(X,Y)=[\zeta_{\rho}^{s}(X,Y)+\zeta_{\rho}^{s}(Y,X)]/2$, and hence $\operatorname{Re} \zeta_{\rho}^{s}(X,Y)$ is symmetric in $X$ and $Y$. Moreover, the bilinear function $\operatorname{Re} \zeta_{\rho}^{s}(X,Y)$ is semi-definite since $\operatorname{Re}\zeta_{\rho}^{s}(X,X)=\zeta_{\rho}^{s}(X,X)\geq 0,\forall X\in \mathcal{H}_{d}$. The above arguments imply the Cauchy-Schwarz inequality, $I^{s}(\rho,X)I^{s}(\rho,Y) \geq |\operatorname{Re} \zeta_{\rho}^{s}(X,Y)|^2$, and that the polarization identity in bilinear form of $\operatorname{Re} \zeta_{\rho}^{s}(X,Y)$ will certainly yield $\operatorname{Re} \zeta_{\rho}^{s}(X,Y)= \frac{1}{4}(I^{s}(\rho,X+Y)-I^{s}(\rho,X-Y))$.
\end{proof}
\begin{theorem}
For a qubit system, we have the following skew information inequalities:
\begin{align}
|\operatorname{Re} \zeta_{\rho}^{s}(X,Y)|^2 \leq I^{s}(\rho,X)I^{s}(\rho,Y) \leq |\zeta_{\rho}^{s}(X,Y)|^2 \; .
\end{align}
Here,
\begin{align}
&\operatorname{Re} \zeta_{\rho}^{s}(X,Y) = \frac{1}{4}(I^{s}(\rho,X+Y)-I^{s}(\rho,X-Y)) \; , \\
&\operatorname{Im} \zeta_{\rho}^{s}(X,Y) = \frac{1}{2i}\tr[\rho[X,Y]] \; .
\end{align}
\end{theorem}
\begin{proof}
The first half part of the inequality has been proved in \cref{theo:gene_qfi_ineq}. For the latter, we proceed the proving under two situations.

Case 1: when $\rho$ is a maximally mixed state, i.e. $\tr[\rho^2]=\frac{1}{2}$ or a pure state, i.e. $\tr[\rho^2]=1$. If $\rho$ is the maximally mixed, the inequality is true on account of $I^{s}(\rho,X)=0$. Otherwise, if $\rho$ is a pure state and $\rho=\sum_{i=1}^{2}\lambda_{i}\ket{\psi_{i}}\bra{\psi_{i}},\lambda_{1}=1,\lambda_{2}=0,\braket{\psi_{i}|\psi_{j}}=\delta_{ij}$, we have
\begin{align}
\zeta_{\rho}^{s}(X,Y) =& \sum_{i\neq j}\left[\lambda_{i}-m_{s}(\lambda_{i},\lambda_{j})\right]\braket{\psi_{i}|X^{\dagger}|\psi_{j}}\braket{\psi_{j}|Y|\psi_{i}} \notag \\
=& \braket{\psi_{1}|X^{\dagger}|\psi_{2}}\braket{\psi_{2}|Y|\psi_{1}} \; .
\end{align}
Here, we have employed the previous convention $m_{s}(1,0)=m_{s}(0,1)=0$. Similarly, we have $I^{s}(\rho,X)=|\braket{\psi_{1}|X^{\dagger}|\psi_{2}}|^2$ and $I^{s}(\rho,Y)=|\braket{\psi_{1}|Y^{\dagger}|\psi_{2}}|^2$.
Thus
\begin{align}
|\zeta_{\rho}^{s}(X,Y)|^2 = I^{s}(\rho,X)I^{s}(\rho,Y) \; .
\end{align}

Case 2: when $\rho$ is a mixed state with $\frac{1}{2}<\tr[\rho^2]<1$. We define the following matrix similar to covariance matrix for two observables $X$ and $Y$:
\begin{align}
\mathcal{P}^{s}(\rho) :=
\left(
\begin{matrix}
\zeta_{\rho}^{s}(X,X) & \zeta_{\rho}^{s}(X,Y) \\
\zeta_{\rho}^{s}(Y,X) & \zeta_{\rho}^{s}(Y,Y)
\end{matrix}
\right) \; .
\end{align}
$\mathcal{P}^{s}(\rho)$ is a Hermitian matrix due to the Hermiticity of $\zeta_{\rho}^{s}(X,Y)$. Now, the question is equivalent to prove that $\mathcal{P}^{s}(\rho)$ is an indefinite matrix, i.e. $\det \mathcal{P}^{s}(\rho)\leq 0$ for arbitrary observables $X,Y$ and quantum state $\rho$. A matrix is indefinite if and only if $\braket{\psi|A|\psi}$ is real for all $\ket{\psi}\in \mathds{C}^2$ and there are vectors $\ket{\psi},\ket{\phi}\in \mathds{C}^2$ such that $\braket{\phi|A|\phi}<0<\braket{\psi|A|\psi}$. Assuming $\ket{\psi}=(c_{1},c_{2})^{\mathrm{T}}$ with $c_{1},c_{2}\in \mathds{C}$, we then have
\begin{align}
\braket{\psi|\mathcal{P}^{s}(\rho)|\psi} = \zeta_{\rho}^{s}(X',X') = I^{s}(\rho,X') \; ,
\end{align}
where $X'=c_{1}X+c_{2}Y$. Therefore, whether $\mathcal{P}^{s}(\rho)$ is indefinite or not hinges completely on the vectors within  subspace $\operatorname{span} \{X,Y\}=\{X'|X'=c_{1}X+c_{2}Y\}$. To be specific, $\mathcal{P}^{s}(\rho)$ is indefinite if and only if there are both positive vector $X'_{1}$ and negative vector $X'_{2}$ within the subspace $\operatorname{span}\{X,Y\}$. Then we may construct  positive and negative vectors when $-\infty<s\leq0$. In light of \cref{th:gene_skew_non}, when $X'$ is a Hermitian matrix, i.e. $c_{1}$ and $c_{2}$ are both real numbers, it is then a positive vector, $I^{s}(\rho,X')>0$. That means for arbitrary observables $X,Y$, the positive vector within the subspace $\operatorname{span}\{X,Y\}$ is obtained. For negative vector, $\forall A\in M_{d}(\mathds{C})$, we have
\begin{align}
I^{s}(\rho,A)
= &(\lambda_{1}-m_{s}(\lambda_{1},\lambda_{2})|\widetilde{A}_{21}|^2 + \notag \\
&(\lambda_{2}-m_{s}(\lambda_{2},\lambda_{1}))|\widetilde{A}_{12}|^2 \; .
\label{eq:qubit_gene_qfi}
\end{align}
Here, $\widetilde{A}_{12}=\braket{\psi_{1}|A|\psi_{2}}$ and $\widetilde{A}_{21}=\braket{\psi_{2}|A|\psi_{1}}$ can be viewed as matrix elements of $A$ with respect to basis $\{\ket{\psi_{1}},\ket{\psi_{2}}\}$, which means $A$ and $\widetilde{A}$ are associated via unitary transformation, $\widetilde{A}=U^{\dagger}AU$ with $UU^{\dagger}=\mathds{1}$.

For qubit system, arbitrary two observables $X,Y$ can be expressed as
\begin{align}
X = \left(
\begin{matrix}
\times & \alpha  \\
\alpha^{*} & \times
\end{matrix}
\right) \; , \;
Y = \left(
\begin{matrix}
\times & \beta  \\
\beta^{*} & \times
\end{matrix}
\right) \; .
\end{align}
Here, the diagonal elements can be any real numbers; $\alpha,\beta$ are arbitrary nonzero complex numbers. Since unitary transformation does not change the Hermiticity, that means after the unitary transformation, ${X}$ and ${Y}$ may take similar forms
\begin{align}
\widetilde{X} = \left(
\begin{matrix}
\times & \alpha'  \\
\alpha'^{*} & \times
\end{matrix}
\right) \; , \;
\widetilde{Y} = \left(
\begin{matrix}
\times & \beta'  \\
\beta'^{*} & \times
\end{matrix}
\right) \; .
\end{align}
And hence
\begin{align}
\widetilde{X}' = U^{\dagger}X'U = \left(
\begin{matrix}
\times & c_{1}\alpha' + c_{2}\beta'  \\
c_{1}\alpha'^{*} + c_{2}\beta'^{*} & \times
\end{matrix}
\right) \; ,
\end{align}
where $X'=c_{1}X+c_{2}Y$. Considering \cref{eq:qubit_gene_qfi}, if $c_{1}$ and $c_{2}$ satisfy $c_{1}\alpha' + c_{2}\beta' \neq 0$ and $c_{1}\alpha'^{*} + c_{2}\beta'^{*} = 0$, then $I^{s}(\rho,X')<0$. Thus the negative vector $X'=c_{1}X+c_{2}Y$ is constructed, where $c_{1}\neq 0$, $c_{2}=-\frac{\alpha'^*}{\beta'^*}c_{1}$ and $\alpha'\neq\kappa \beta'$ with $\kappa$ being a real number. Note, in case $\alpha'=0$ or $\beta'=0$, we readily have $I^{s}(\rho,X)=0$ or $I^{s}(\rho,Y)=0$. On the other hand, if $\alpha'=\kappa \beta'$ and $X$ is linearly dependent on $Y$, i.e., $X=\kappa Y$, then $\det \mathcal{P}^{s}(\rho)=0$. Of the case $s=-\infty$, because $I^{-\infty}(\rho,X)
= (\lambda_{1}-\lambda_{2})|\braket{\psi_{1}|X|\psi_{2}}|^2$ and $\zeta_{\rho}^{-\infty}(X,Y) = (\lambda_{2}-\lambda_{1})\braket{\psi_{1}|X|\psi_{2}}\braket{\psi_{1}|Y|\psi_{2}}$,
we have $I^{-\infty}(\rho,X)I^{-\infty}(\rho,Y)=|\zeta_{\rho}^{-\infty}(X,Y)|^2$. Hereto, we have proved the theorem in all possible situations, and notice that the polarization identity for the sesquilinear form of $\zeta_{\rho}^{s}(X,Y)$ gives
\begin{align}
\zeta_{\rho}^{s}(X,Y) = &\frac{1}{4}[(I^{s}(\rho,X+Y)-I^{s}(\rho,X-Y)) + \notag \\
&i(I^{s}(\rho,X-iY)-I^{s}(\rho,X+iY))] \; ,
\end{align}
which implies
\begin{align}
&\operatorname{Re} \zeta_{\rho}^{s}(X,Y) = \frac{1}{4}(I^{s}(\rho,X+Y)-I^{s}(\rho,X-Y)) \; , \\
&\operatorname{Im} \zeta_{\rho}^{s}(X,Y) = \frac{1}{2i}\tr[\rho[X,Y]] \; .
\end{align}
\end{proof}
\begin{corollary}\label{coro:gene_skew_inf_inequality}
For a qubit system, the skew information inequality
\begin{align}
\mathcal{I}_{\rho}^{s}(X,Y) \leq \frac{1}{4}|\tr[\rho [X,Y]]|^2 \;
\label{eq:gene_skew_inf_inequality}
\end{align}
exists, where $\mathcal{I}_{\rho}^{s}(X,Y):=I^{s}(\rho,X)I^{s}(\rho,Y)-
\frac{1}{16}[I^{s}(\rho,X+Y)-I^{s}(\rho,X-Y)]^2$ and $\mathcal{I}_{\rho}^{s}(X,Y)\geq 0$.
\end{corollary}
\begin{figure}
\centering
\includegraphics[width=0.8\linewidth]{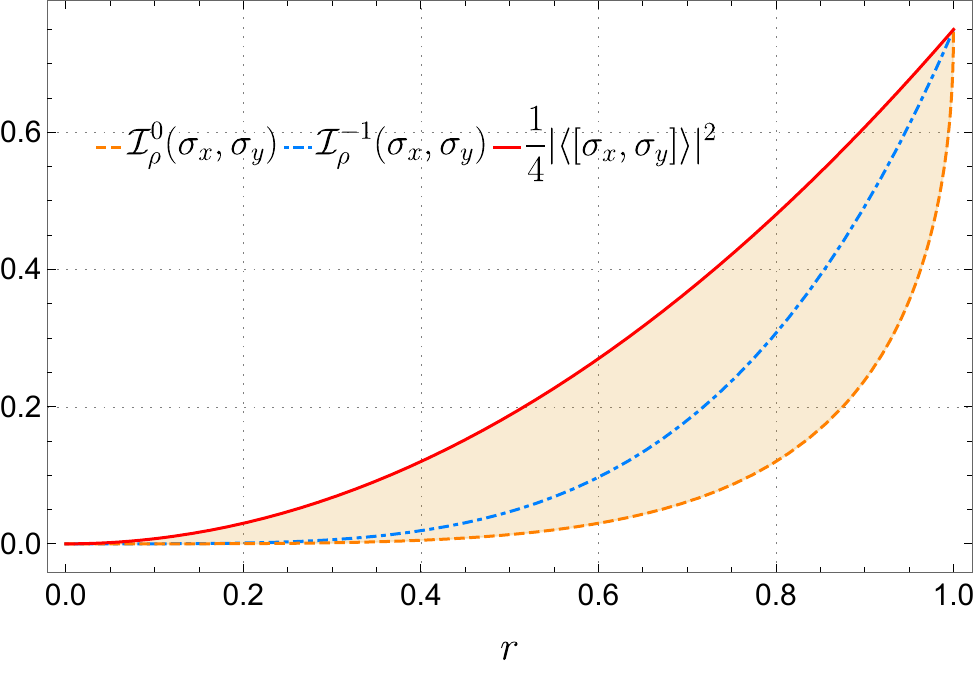}
\caption{\label{fig:gene_skew_inf_inequality} The skew information inequality for Pauli operators $\sigma_{x}$ and $\sigma_{y}$ in quantum state $\rho=1/2(\mathds{1}+\vec{r}\cdot\vec{\sigma})$ with $\vec{r}=r(\sin\theta\cos\phi,\sin\theta\sin\phi,\cos\theta)^{\mathrm{T}},\theta=\phi=\pi/6,r\in[0,1]$. The dashdotted blue and dashed orange curves correspond to quantum Fisher information and Wigner-Yanase skew information respectively.}
\end{figure}
For comparison, we notice that the Schrödinger uncertainty relation \cref{eq:schrodinger_form} can be reformulated as
\begin{align}
\mathcal{V}_{\rho}(X,Y) \geq \frac{1}{4}|\tr[\rho [X,Y]]|^2 \; ,
\label{eq:schrodinger_form1}
\end{align}
with $\mathcal{V}_{\rho}(X,Y):=V(\rho,X)V(\rho,Y)-
\frac{1}{16}[V(\rho,X+Y)-V(\rho,X-Y)]^2$. \cref{eq:schrodinger_form1} is an improvement to the Robertson uncertainty relation \cite{robertson29}
\begin{align}
V(\rho,X)V(\rho,Y) \geq \frac{1}{4}|\tr[\rho [X,Y]]|^2 \; .
\end{align}
\cref{eq:schrodinger_form1} sets a lower bound for the measurement uncertainty of incompatible observables. As a counterpart, our inequality \cref{eq:gene_skew_inf_inequality} provides an upper bound for accessible information about incompatible observables. To illustrate, we plot the constraint contour for the case of Pauli operators $\sigma_{x}$ and $\sigma_{y}$ in \cref{fig:gene_skew_inf_inequality}.

Now we exhibit certain physical meanings of the above results. The real part of $\zeta_{\rho}^{s}(X,Y)$ can be reformulated with the anticommutator of $X,Y$, and we then get the following inequality:
\begin{align}
\frac{1}{4}\left|\Braket{\{X,Y\}}-2\operatorname{Re}\eta_{\rho}^{s}(X,Y)\right|^2& \leq I^{s}(\rho,X)I^{s}(\rho,Y) \notag \\
\leq \frac{1}{4}\left|\Braket{[X,Y]}\right|^2& + \frac{1}{4}\left|\Braket{\{X,Y\}}-2\operatorname{Re}\eta_{\rho}^{s}(X,Y)\right|^2 \; .
\label{eq:coherence_ur}
\end{align}
Here, $\eta_{\rho}^{s}(X,Y):=\sum_{i,j}m_{s}(\lambda_{i},\lambda_{j})\braket{\psi_{i}|X|\psi_{j}}\braket{\psi_{j}|Y|\psi_{i}}$. It has been demonstrated that the skew information can be viewed as a measure of quantum coherence \cite{girolami14} and satisfies the bona fide
criteria for coherence monotones \cite{baumgratz14}. Therefore, the generalized skew information provide the more alternative choices for coherence measure. The coherence of a quantum state depends on the choice of the reference basis. An interesting question is to consider the trade-off relation for quantum coherence measures in different reference bases \cite{cheng15,singh16,fan19}. $I^{s}(\rho,X)$ and $I^{s}(\rho,Y)$ quantifies the coherence in the eigen-basis of observables $X$ and $Y$, respectively. Therefore, \cref{eq:coherence_ur} also represents a coherence uncertainty relation.


\section{Discussion}

The generalized skew information is found can play an important role in quantum information science. In this paper, we revisit the skew information introduced by Wigner and Yanase and define a generalized skew information quantity $I^{s}(\rho,X)$ which contains Wigner-Yanase skew information and QFI as special cases. We give a transparent convexity proof about $I^{s}(\rho,X)$, which in the meantime provides a simple proof of Wigner-Yanase-Dyson conjecture. We consider the quantum limit from the information point of view, which yields the skew information inequality and stands as a supplements to the known uncertainty relation. In practice, for accessible information content, the upper limit usually might be more meaningful than the lower one. Nevertheless, it should be noted that the result \cref{eq:gene_skew_inf_inequality} is only applicable to the qubit system, the highlight in quantum information, its high dimensional extension still requires extra strength to fulfill.

\section*{Acknowledgements}
\noindent
This work was supported in part by the National Natural Science Foundation of China(NSFC) under the Grants 11975236 and by the University of Chinese Academy of Sciences.





\end{document}